\documentclass[10pt, conference, letterpaper]{IEEEtran}
\IEEEoverridecommandlockouts
\usepackage{cite}
\usepackage{amsmath,amssymb,amsfonts}
\usepackage{algorithmicx}
\usepackage[noend,ruled, linesnumbered, vlined]{algorithm2e}
\usepackage{array}
\usepackage{graphicx}
\usepackage{stfloats}
\usepackage{amsthm}
\usepackage{subfigure}
\usepackage{url}
\usepackage{verbatim}
\usepackage{algpseudocode}  
\usepackage{textcomp}
\usepackage{xcolor}
\usepackage{color}
\newenvironment{sequation}{\begin{equation}\small}{\end{equation}}

\newtheorem{lemma}{\textbf{Lemma}}
\newtheorem{theorem}{\textbf{Theorem}}

\newtheorem{definition}{\textbf{Definition}}
\newtheorem{corollary}{\textbf{Corollary}}

\usepackage{multirow}
\usepackage{bm}
\usepackage{booktabs}
\usepackage{epstopdf}
\usepackage{cases}
\usepackage{makecell,multirow,diagbox}
\usepackage{stfloats}
\usepackage{comment}
\usepackage{array}

\usepackage{enumitem}
\definecolor{b}{rgb}{1, 1, 1}
\definecolor{r}{rgb}{1, 1, 1}
\def\BibTeX{{\rm B\kern-.05em{\sc i\kern-.025em b}\kern-.08emT\kern-.1667em\lower.7ex\hbox{E}\kern-.125emX}}

\begin{document}

\title{A Two Time-Scale Joint Optimization Approach for UAV-assisted MEC}

\author{\IEEEauthorblockN{Zemin Sun\IEEEauthorrefmark{2},
		Geng Sun\IEEEauthorrefmark{2}\IEEEauthorrefmark{1},  
        Long He\IEEEauthorrefmark{2},
        Fang Mei\IEEEauthorrefmark{2},
	  Shuang Liang\IEEEauthorrefmark{3},
	  Yanheng Liu\IEEEauthorrefmark{2}
	}
\IEEEauthorblockA{\IEEEauthorrefmark{2}College of Computer Science and Technology, Jilin University, Changchun 130012, China} 
 \IEEEauthorblockA{\IEEEauthorrefmark{3}School of Information Science and Technology, Northeast Normal University, Changchun 130012, China} 
 E-mails: \{sunzemin, sungeng\}@jlu.edu.cn, helong0517@foxmail.com, meifang@jlu.edu.cn, \\ liangshuang@nenu.edu.cn, yhliu@jlu.edu.cn\\
\IEEEauthorrefmark{1}Corresponding author: Geng Sun}

\maketitle

\begin{abstract}		
	Unmanned aerial vehicles (UAV)-assisted mobile edge computing (MEC) is emerging as a promising paradigm to provide aerial-terrestrial computing services close to mobile devices (MDs). However, meeting the demands of computation-intensive and delay-sensitive tasks for MDs poses several challenges, including the demand-supply contradiction between MDs and MEC servers, the demand-supply heterogeneity between MDs and MEC servers, the trajectory control requirements on energy efficiency and timeliness, and the different time-scale dynamics of the network. To address these issues, we first present a hierarchical architecture by incorporating terrestrial-aerial computing capabilities and leveraging UAV flexibility. Furthermore, we formulate a joint computing resource allocation, computation offloading, and trajectory control problem to maximize the system utility. Since the problem is a non-convex mixed integer nonlinear programming (MINLP), we propose a two time-scale joint computing resource allocation, computation offloading, and trajectory control (TJCCT) approach. In the short time scale, we propose a price-incentive method for on-demand computing resource allocation and a matching mechanism-based method for computation offloading. In the long time scale, we propose a convex optimization-based method for UAV trajectory control. Besides, we prove the stability, optimality, and polynomial complexity of TJCCT. Simulation results demonstrate that TJCCT outperforms the comparative algorithms in terms of the utility of the system, the QoE of MDs, and the revenue of MEC servers.

\end{abstract}


\maketitle

%
%

\section{Introduction}
\label{sec_introduction}

\par  \IEEEPARstart{T}{he} development of wireless technologies and the proliferation of mobile devices (MDs) trigger various emerging applications, 
such as real-time video analysis~\cite{hou2023eavs}, online gaming, and augmented reality~\cite{Xu2023learn}. These applications often require extensive computing resources and low latency for the quality of experience (QoE)~\cite{dai2022bloom}. However, fulfilling the computation-intensive and delay-sensitive computation tasks of these applications poses a great challenge to MDs with insufficient computing capability. To tackle this challenge, mobile edge computing (MEC) has been identified as a promising technology to meet the stringent requirements of these applications. By offloading the computation-intensive tasks to proximate MEC servers, the QoE of MDs can be significantly enhanced in a cost-effective and energy-efficient way \cite{li2023tapfinger,sun2023bar}. However, due to the dependence on terrestrial infrastructures and the environment, conventional terrestrial MEC servers are limited by the high cost of deployment, low adaptability to the network dynamic, and fixed service range. 

\par Recent years have seen a paradigm shift from terrestrial edge computing toward aerial-terrestrial edge computing, i.e., UAV-assisted MEC~\cite{qu2023elastic,dong2021uavs}. With high maneuverability, UAVs could be rapidly and flexibly deployed as aerial MEC servers to assist the terrestrial MEC servers in providing temporary computing services whenever and wherever needed. Furthermore, the high-probability line-of-sight (LoS) link of UAVs can improve the communication reliability and network capacity of the terrestrial MEC networks~\cite{Li2023TMC,Qu2021}.

\par Despite the aforementioned benefits, the UAV-assisted MEC is facing some unprecedented challenges. \textit{\textbf{i) Demand-Supply Contradiction for Resource Allocation.}} Compared to the cloud, MEC servers have limited computing capabilities, particularly for aerial MEC servers with constrained carrying capacity. However, the computation tasks of MDs are often computation-hungry and latency-sensitive. This demand-supply contradiction between the limited computing resources of MEC servers and the stringent requirement of MDs poses a challenge for efficient computing resource allocation. \textit{\textbf{ii) Demand-Supply Heterogeneity for Computation Offloading.}} Different computation tasks of MDs have diverse requirements on computing resources, while different MEC servers possess varying computing capabilities. This demand-supply heterogeneity between the computation tasks of MDs and MEC servers could incur resource under-utilization among MEC servers, which brings difficulties in designing efficient computation offloading methods to ensure satisfied QoE for MDs and high resource utilization among MEC servers. \textit{\textbf{iii) Energy-Efficient and Real-Time Trajectory Control.}} The mobility of MDs and random generation of computation tasks lead to spatiotemporal dynamics in the offloading requirements, which necessitates real-time trajectory control. However, the intrinsic limited onboard energy of UAVs restricts the service time, thus posing challenges for energy-efficient and real-time  UAV trajectory control. \textit{\textbf{iii) Different Time-Scale Dynamics.}} The dynamic characteristics of the UAV-assisted MEC network, such as the dynamic of the channel, random arrival of tasks, and mobility of MDs, vary across different time scales. Accordingly, integrating these features into a joint optimization framework for computing resource allocation, computation offloading, and trajectory control is of great significance but leads complexity to the algorithm design.

\par This work presents a two time-scale computing resource allocation, computation offloading, and UAV trajectory control approach for the UAV-assisted MEC system. The main contributions are as follows:

\begin{itemize}
	\item \textit{\textbf{System Architecture.}} We employ a hierarchical architecture for a UAV-assisted MEC system that consists of a UE layer, a terrestrial edge layer, an aerial edge layer, and a control layer. Under the coordination of the software-defined networking (SDN) controller, the two time-scale decisions are made to deal with the demand-supply contradiction between MDs and MEC servers, demand-supply heterogeneity between MDs and MEC servers, and the spatiotemporal dynamics of tasks.
	
	\item \textit{\textbf{Problem Formulation.}} We formulate a joint computing resource allocation, computation offloading, and trajectory control problem to maximize the system utility.

    \item \textit{\textbf{Algorithm Design.}}  To solve the formulated problem, we propose a two time-scale joint computing resource allocation, computation offloading, and trajectory control (TJCCT) algorithm. Specifically, TJCCT incorporates two time-scale optimization methods. In the short time scale, we propose a price-incentive method for on-demand computing resource allocation and a matching mechanism-based method for computation offloading. In the long time scale, we propose a convex optimization-based method for UAV trajectory control.

	\item \textbf{\textit{Performance Evaluation.}} The performance of TJCCT is verified through both theoretical analysis and simulations. First, the optimality and polynomial complexity of TJCCT is proved theoretically. Furthermore, simulation results demonstrate that TJCCT outperforms the comparative algorithms in terms of overall system performance.
\end{itemize}

\par The remaining of this work is organized as follows. Section \ref{sec_related work} reviews the related work. Section \ref{sec_model} presents the system models and problem formulation. Section \ref{sec_jointOffloading} elaborates on the proposed TJCCT. The theoretical analysis is given in Section \ref{sec_stability}. Section \ref{sec_simulation} shows the simulation results and discussions. Finally, this work is concluded in Section \ref{sec_conclusion}.

\section{Related work}
\label{sec_related work}

\par \textbf{\textit{Joint Computation Offloading and Resource Allocation.}} Most existing works focus on joint computation offloading and computing resource allocation. For example, Yu et al. \cite{Yu2020} proposed a joint computation offloading and resource allocation approach to minimize the weighted sum of the service delay and UAV energy consumption. Ding et al. \cite{Ding2023} investigated the offloading and resource management for UAV-assisted MEC secure communications. Nie et al. \cite{nie2021semi} presented a semi-distributed method for joint computation offloading and resource allocation. Guo et al. \cite{10134570} explored the problem of joint task scheduling and computing resource allocation with computation data dependency. Goudarzi et al. \cite{10058597} utilized cooperative evolutionary computation to solve the joint optimization of computation offloading and computing resource allocation. However, in most of these studies, UAVs are either fixed or follow predetermined trajectories, which may not be suitable for the scenarios with randomly-distributed MDs.

\par \textbf{\textit{Joint Computation Offloading and Trajectory Control.}} To harness the full potential of flexible and elastic offloading services, joint computation offloading and UAV trajectory control has attracted widespread attention. For example, Lin et al. \cite{lin2023pddqnlp} jointly optimized the computation offloading and UAV trajectory to maximize the energy efficiency of UAVs. Shen et al. \cite{9796968} proposed a QoE-oriented service provision model to optimize the UAV trajectory and computing resource allocation. Chen et al. \cite{Chen2023} focused on the maximum-minimum average secrecy capacity by jointly optimizing the trajectory and computation offloading. Wang et al. \cite{Wang2022a} proposed a cooperative offloading scheme for post-disaster rescue. However, most of these studies focused on a relatively static scenario with stationary users and employed the offline algorithm to plan the entire trajectory across the system timeline. Therefore, these works may not be directly applicable to our study with multiple UAVs and mobile users.

\par This work studies the joint computing resource allocation, computation offloading, and UAV trajectory control in the UAV-assisted MEC system, where the limited computing resources of MEC servers, the stringent requirements of computation tasks, the heterogeneity between MDs and MEC servers, and the dynamics of the network are jointly considered.

%
%
\section{System Model and Problem Formulation}
\label{sec_model}
\subsection{System Model}
\label{sec_system_model}

\subsubsection{System Overview}
\label{sec:system_overview}

\begin{figure}[!t] 
	\centering
	\includegraphics[width =3.5in]{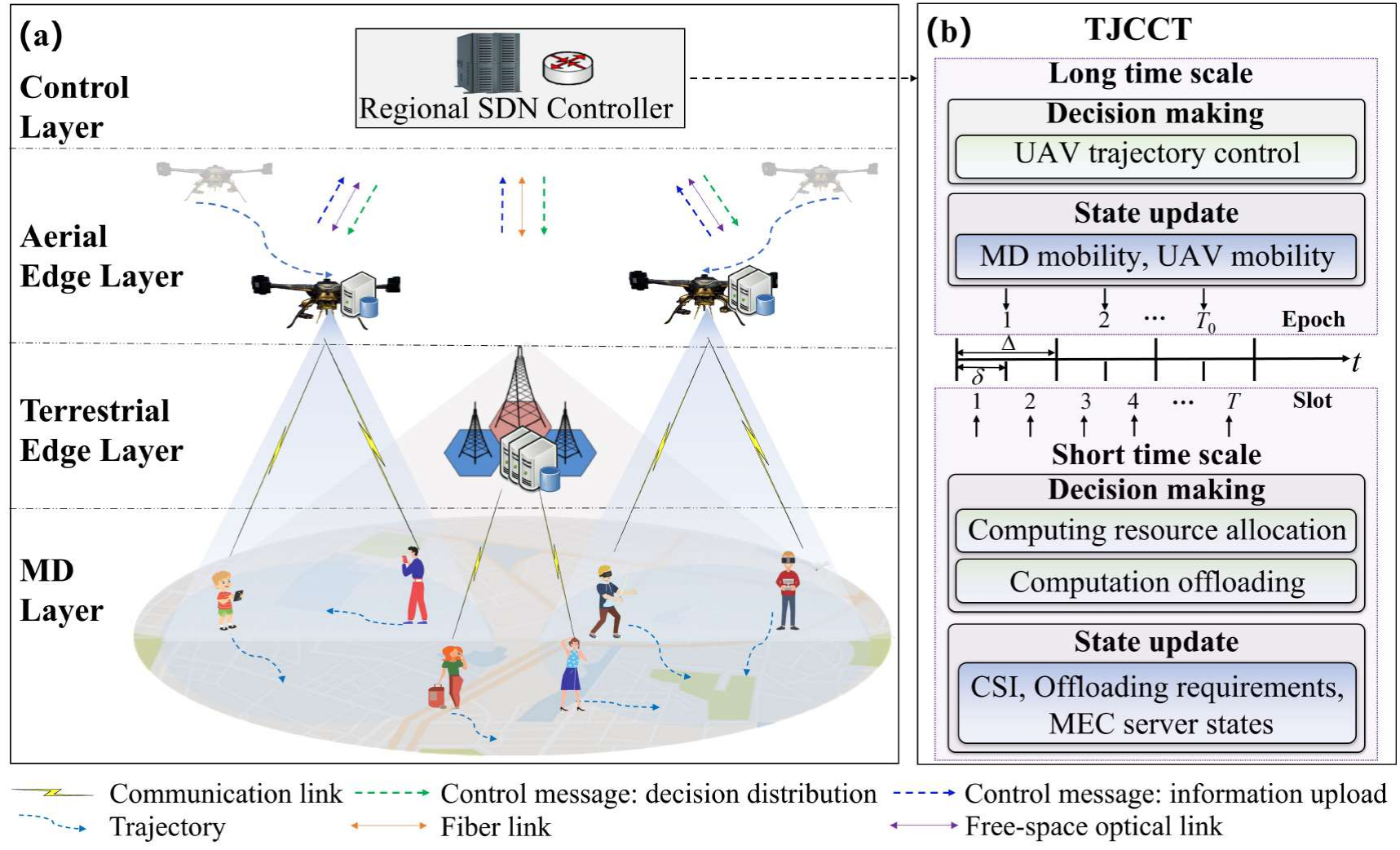}
	\caption{The architecture of computing resource allocation, computation offloading, and UAV trajectory control for UAV-assisted MEC system.}
	\label{fig_systemModel}
\end{figure}

\par We consider a hierarchical UAV-assisted MEC system as shown in Fig. \ref{fig_systemModel}.

\textbf{\textit{In the spatial dimension}}, the hierarchical UAV-assisted MEC system comprises an MD layer, a terrestrial edge layer, an aerial edge layer, and a control layer. Specifically, \textit{at the MD layer}, a set of MDs $\mathcal{I}=\{1, \ldots, I\}$ moving in the considered area periodically handle the tasks with diverse requirements. \textit{At the terrestrial edge layer}, one macro base station (MBS) $b$ equipped with terrestrial MEC server provides edge computing services for the MDs within its service range. The MBS is connected to the control layer through high-speed fiber links \cite{Wang2022}. \textit{At the aerial edge layer}, the rotary-wing UAVs $\mathcal{U}=\{1, \ldots, U\}$ equipped with aerial MEC servers are dispatched as aerial base stations to assist the MBS in providing supplementary computing services for MDs. The UAVs are connected to the control layer through free-space optical links. \textit{At the control layer}, the regional SDN controller, on which our algorithm runs, coordinates the decisions regarding computation offloading for MDs, computing resource allocation for MEC servers, and trajectory control for UAVs based on the knowledge acquired from the terrestrial and aerial edge layers. Furthermore, we consider that the radio access links of the system are allocated orthogonal frequency bands \cite{ji2022trajectory}. Additionally, terrestrial MEC server and MBS are are used interchangeably, and the aerial MEC server and UAV are used interchangeably. Besides, the terrestrial MEC server and aerial MEC servers are collectively referred to MEC servers which are indexed as $j\in\{b\cup \mathcal{U}\}$.

\par \textbf{\textit{In the temporal dimension}}, we consider that the system operates in a two time-scale manner since the dynamic characteristics for channel state, task arrival of MDs, and workload update of edge servers vary in a fine-grained timescale, while the mobility of MDs varies in a long timescale \cite{Li2021}. Specifically, the finite system time horizon is discretized into $T$ time slots $\mathcal{T}=\{1,\ldots,T\}$ with equal slot duration $\delta$, which is consistent with the coherence block of the wireless channel \cite{Liao2021, Gou2023achieve}. Furthermore, every $\Delta$ consecutive slots are combined into a time epoch indexed by  $t_0\in \mathcal{T}_0=\{1,\ldots,T_0\}$ ($\Delta$ is selected to be sufficiently small to guarantee that the positions of UAVs and MDs are approximately constant within each epoch). Therefore, in the short timescale, the channel state information, offloading requirements of MDs, and states of MEC servers are captured and updated, and the decisions of task offloading and resource allocation are determined. In the long timescale, the mobility states of MDs are captured and updated, and the UAV trajectory control is decided.

\subsubsection{Basic Models}
	
\par The basic models are given as follows.

\par \textit{\textbf{MD Mobility Model.}} The horizontal coordinate of MD $i$ in time slot $t$ is denoted as $\mathbf{q}_{i,t}=[x_{i}^t,y_{i}^t]_{i\in\mathcal{I}}$. Moreover, we adopt Gauss Markov model \cite{liang1999predictive, Batabyal2015} to capture the temporal-dependent randomness in the movement of MDs. Specifically, the velocity of MD $i\in \mathcal{I}$ at time epoch $t_0+1$ (i.e., time slot $(t_0+1)\Delta$) can be given as
\begin{sequation}
	\label{eq_MD_mobility}
		 \mathbf{v}_{i}^{(t_0+1)\Delta}=\alpha\cdot  \mathbf{v}_{i}^{t_0\Delta}+\big(1-\alpha\big) \bar{\mathbf{v}}+\bar{\sigma}\sqrt{1-\alpha^2} \mathbf{w}_{i},
\end{sequation}

\noindent where $ \mathbf{v}_i^{t_0\Delta}$ is the velocity vector at time epoch $t_0$ and $ \mathbf{w}_i$ is the uncorrelated random Gaussian process, i.e., $ \mathbf{w}_i\sim f^{\text{Gua}}\big(0,\bar{\sigma}^2\big)$. Besides, $0\leq \alpha\leq 1$, $\bar{\mathbf{v}}$, and $\bar{\sigma}$ denote the memory degree, asymptotic mean, and asymptotic standard deviation of velocity, respectively. Therefore, the location of each MD can be updated as
\begin{equation}
	\label{eq_MD_position}
	\begin{aligned}
		\mathbf{q}_{i}^{(t_0+1)\Delta}=\mathbf{q}_{i}^{t_0\Delta} + \mathbf{v}_i^{t_0\Delta} \cdot\delta\cdot\Delta, \forall i\in \mathcal{I}, \  t_0 \in \mathcal{T}_0.
	\end{aligned}
\end{equation}

\par \textit{\textbf{UAV Mobility Model.}} Similar to most existing studies \cite{Wang2022a,song2022evo}, we consider that each UAV $j$ flies at a fixed altitude $H$ with the instantaneous horizontal coordinate of $\mathbf{q}_{j}^{t}=[x_{j}^t,y_{j}^t]_{j\in\mathcal{U}}$. Therefore, the location of each UAV can be updated as follows:
\begin{equation}
	\label{eq_UAV_position}
	\begin{aligned}
	\mathbf{q}_{j}^{(t_0+1)\Delta}=\mathbf{q}_{j}^{t_0\Delta}+\mathbf{v}_{j}^{t_0\Delta}\cdot\delta\cdot\Delta, \forall j\in \mathcal{U},\ t_0 \in \mathcal{T}_0,
	\end{aligned}
\end{equation}

\noindent where $\mathbf{v}_{j}^{t_0\Delta}$ denotes the velocity of UAV $j$ at time epoch $t_0$. Furthermore, the position of each UAV should satisfy several practical constraints: 
\begin{subequations}
	\label{eq_UAV_mob}
	\begin{alignat}{2}
		&\textbf{q}_{j}^{1}=\textbf{q}_j^{I},\ \textbf{q}_{j}^{T_0\Delta}=\textbf{q}_j^{F}, \ \forall j\in\mathcal{U},\label{eq_UAV_mob_ini_fin}\\
		&||\textbf{q}_{j}^{(t_0+1)\Delta}-\textbf{q}_{j}^{t_0\Delta}||\leq v_{\text{U}}^{\max}\delta\Delta,\ \forall j\in\mathcal{U}, \ t_0 \in \mathcal{T}_0, \label{eq_UAV_mob_dis_each}\\
		&||\textbf{q}_{j}^{F}-\textbf{q}_{j}^{t_0\Delta}||\leq v_{\text{U}}^{\max}  (T_0-t_0)\delta \Delta,\ \forall j\in\mathcal{U}, \ t_0 \in \mathcal{T}_0, \label{eq_UAV_mob_dis_each1}
	\end{alignat}
\end{subequations}


\noindent where $v_U^{\max}$ is the maximum velocity of UAV. Constraint (\ref{eq_UAV_mob_ini_fin}) predetermines the initial and final positions of UAVs. Besides, constraints \eqref{eq_UAV_mob_dis_each} and \eqref{eq_UAV_mob_dis_each1} indicate that 
the trajectory of each UAV is constrained by the maximum velocity.

\par\textit{\textbf{MD Model.}} Each MD $i \in \mathcal{I}$ is characterized by the tuple $\big\langle f_i^{\max}, n_i^{\text{core}}, \tau_i^t, \zeta_{i}^t, \mathcal{K}_{i}^{t}\big\rangle$ in time slot $t$, where $f_i^{\max}$ represents the computing capability (in cycles/s), $n_i^{\text{core}}$ is the number of CPU core, $\tau_i^t$ denotes the energy constraint, $\zeta_{i}^t\in\{0,1\}$ is a binary variable that indicates whether a computing task is generated during time slot $t$, and $ \mathcal{K}_{i}^{t}$ denotes the task generated in time slot $t$. Furthermore, we assume that each MD is equipped with a single CPU core, i.e., $n_i^{\text{core}}=1$ \cite{Dai2019a}. 

\par\textit{\textbf{MEC Server Model.}} Similar to \cite{Dai2019a}, we consider that the MEC servers are equipped with multi-core CPUs to enable parallel processing of multiple tasks. Consequently, each MEC server is characterized by the tuple $\langle n_j^{\text{core}}, f_j^{\max}, E_j^{\max}\rangle$,  where $n_j^{\text{core}}$ denotes the number of CPU cores, and each of these CPU cores assumed to have homogeneous computing resources of $f_j^{\max}$ (cycles/s). Moreover, $E_j^{\max}$ is the energy constraint.

\par\textit{\textbf{Computation Task Model.}} The computation task $\mathcal{K}_{i}^{t}$ is characterized as $ \mathcal{K}_{i}^{t}=\langle l_{i}^{t}, \mu_{i}^t, \tau_{i}^t\rangle$, where $l_{i}^{t}$ is the task size (in bits), $\mu_{i}^t$ is the required computation resources (in cycles), and $\tau_{i}^t$ is the deadline of the task. 

\par \textit{\textbf{Decision Variables.}} The decisions of computing resource allocation, computation offloading, and UAV trajectory control are jointly considered.  \textit{\textbf{i)}} Given that executing tasks at remote MEC servers incurs additional costs that should be covered by MDs, the unit price of computing resources $p_{j,i}^t$ (\$ / GHz) is introduced to capture the costs and revenue involved in the computation offloading process \cite{Xia2021}. Therefore, the computing resource allocation decision is denoted as $\big(f_{j,i}^t, p_{j,i}^t\big), \ \forall j \in \{b, \mathcal{U}\}$, where $f_{j,i}^t$ is the amount of computing resources allocated by MEC server $j$ to MD $i$ in time slot $t$, and $p_{j,i}^t$ is the unit price of computing resources paid by MD $i$ for completing task $\mathcal{K}_i^t$. \textit{\textbf{ii)}} The computation offloading decision is denoted as a binary variable $o_{i,n}^t\in\{0,1\}, n \in \mathcal{N} = \{0,b,\mathcal{U}\}$, which indicates task $\mathcal{K}_i^t$ of MD $i$ is processed locally ($o_{i,0}^t=1$) or offloaded to MEC server $j\in\{b,\mathcal{U}\}$ ($o_{i,j}^t=1$) in time slot $t$. \textit{\textbf{iii)}} The trajectory for each UAV is denoted as $\{\mathbf{q}_j^t\}_{j\in \mathcal{U}, t\in \mathcal{T}}$.

%
%
\subsection{Communication Model}
\label{sec_communicationModel}

\par By using the widely used orthogonal frequency division multiple access (OFDMA) \cite{ji2022trajectory}, the uplink data rate between MD $i\in\mathcal{I}$ and MEC server  $j\in\{b\cup \mathcal{U}\}$ can be given as:
\begin{sequation}
	\label{eq_dataRate}
	r_{i,j}^t=B_{i,j}^t\cdot \log_2\big(1+P_{i}^{t}\cdot g_{i,j}^t/N_0\big),
\end{sequation}

\noindent where $B_{i,j}^t$ is the subchannel bandwidth between MD $i$ and MEC server $j$ in time slot $t$, $P_{i}^{t}$ is the instantaneous transmit power of MD $i$, $N_0$ is the noise power, and $g_{i,j}^t$ is the instantaneous channel power gain.

\par Due to the complex nature of the communication environment in UAV-assisted MEC networks, such as the movement of MDs and occasional blockages caused by obstacles, the channel power gain is calculated by incorporating the commonly used probabilistic LoS channel with the large-scale and small-scale fadings as $g_{i,j}^t={\mathbb{P}_{i,j}^{t}}\cdot g_{i,j}^{t,\text{L}} +(1-\mathbb{P}_{i,j}^{t})\cdot  g_{i,j}^{t,\text{N}}$~\cite{ji2022trajectory}, where $g_{i,j}^{t,x}$ denotes the channel power gain between MD $i$ and MEC server $j$, $\mathbb{P}_{i,j}^{t}$ denotes the LoS probability, and $ x\in\{\mathrm{L},\mathrm{N}\}$ represents LoS or non-line-of-sight (NLoS). The parameters ${\mathbb{P}_{i,j}^{t}}$ and $g_{i,j}^{t,x}$ are detailed as follows.

\par  \textbf{LoS Probability.} For the MD-MBS link, the LoS probability can be calculated referring to 3GPP standard~\cite{3GPP3D2015} as $\mathbb{P}_{i,j}^t=\min (d_1/d_{i,j}^t, 1)(1-\exp(-d_{i,j}^t/d_2))+\exp({-d_{i,j}^t/d_2})$ ($j=b$),  where $d_{i,j}^t$ is the instantaneous distance between MD $i$ and MEC server $j$, $d_1$ and $d_2$ are parameters to fit the specific scenarios.  For the MD-UAV link, an extensively employed LoS probability is given as \cite{AlHourani2014,sun2023uav} $\mathbb{P}_{i,j}^t=1/(1+\varrho_1 \cdot \exp {(-\varrho_2(\frac{180}{\pi}\arctan(\frac{H}{d_{i,j}^{t}})-\varrho_1))})$($j\in \mathcal{U}$), where $d_{i,j}^t=\|\mathbf{q}_i^t-\mathbf{q}_j^t\|$ is the horizontal distance between MD $i$ and UAV $j$, and $\varrho_1$ and $\varrho_2$ are environment-dependent parameters.

\par  \textbf{Channel Power Gain.} According to \cite{3GPPTR389012020}, the channel power gain between MD $i$ and MEC server $j\in\{b,\mathcal{U}\}$ in time slot $t$ can be uniformly given as $g_{i,j}^{t,x}=|h_{i,j}^{t,x}|^2(L_{i,j}^{t,x})^{-1}$, where $h_{i,j}^{t,x}$ and $L_{i,j}^{t,x}$ denote the parameters of small-scale fading and large-scale fading for LoS and NLoS links, which are modeled as the parametric-scalable Nakagami-$m$ fading \cite{cheng2007mobile} and the shadow fading \cite{3GPPTR389012020}, respectively.

%
%

\subsection{Computation Model}
\label{sec_DelayModel}

\par The computation offloading decision generally incurs overheads in terms of delay and energy consumption, which are detailed as follows.

\subsubsection{Local Computing Model}

\par The completing delay and energy consumption for local computing are given as follows.
 
\par \textbf{Completion Delay.} When task $\mathcal{K}_{i}^{t}$ is executed locally by MD $i$, the task completion delay is mainly incurred by task computation, which can be given as: 
\begin{sequation}
 	\label{eq_time_local}
 	D_{i,i}^t=\mu_{i}^t/f_i^t,
 \end{sequation}

 \noindent where $f_i^t$ is the currently available computing resources of $i$.
 
\par \textbf{Energy Consumption.} Correspondingly, the energy consumption of MD $i$ to execute task $\mathcal{K}_i^t$ locally is given as:
\begin{sequation}
	\label{eq_energyLocalExe}
	E_{i,i}^t=\gamma_i(f_i^t)^{2}\mu_{i}^t,
\end{sequation}

\noindent where $\gamma_i\geq0$ is the effective capacitance of the CPU  that relies on the CPU chip architecture \cite{pan2021cost}.

\subsubsection{Edge Offloading Model}

\par When task $\mathcal{K}_{i}^{t}$ is offloaded to terrestrial MEC sever $j\in\{b,\mathcal{U}\}$, the completing delay and energy consumption are given as follows.

\par \textbf{Completion Delay.} The task completion delay mainly consists of transmission delay and computation delay, i.e.,
\begin{sequation}
	\label{eq_edge_delay}
	\begin{aligned}
		D_{i,j}^t=l_{i}^{t}/r_{i,j}^t+\mu_{i}^t/f_{j,i}^t,
	\end{aligned}
\end{sequation}

\noindent where $f_{j,i}^t$ denotes the computing resources allocated by MEC server $j$ to MD $i$ in time slot $t$.

\par \textbf{Energy Consumption.} The energy consumption consists of transmission energy for the MD, computation energy for the MEC server, and flight energy for the UAV. For MD $i\in\mathcal{I}$, the energy consumption for task uploading is:
\begin{sequation}
	\label{eq_energyMDMEC}
	E_{i,j}^t=P_{i}^{t} l_{i}^{t}/r_{i,j}^t.
\end{sequation}

\par Furthermore, for MEC server $j\in \{b,\mathcal{U}\}$, the energy consumption for task execution can be given as $E_{j,i}^{t,\text{comp}}=\gamma_j(f_{j,i}^t)^{2}\mu_{i}^t$, where $\gamma_j\geq0$ denotes the effective capacitance of terrestrial MEC server $j$'s CPU. For aerial MEC server $j\in \mathcal{U}$, the energy consumption is incurred not only during computation but also during UAV flight. Specifically, the propulsion power of the rotary-wing UAV in straight-and-level flight includes the components of blade profile power, induced power, and parasite power \cite{Zeng2019,pan2023joint}, which can be given as {\small$	E_{j}^{p}=\underbrace{\eta_1\big(1+3(v_{j}^{t})^2/({v_j^{\text{tip}}})^2)}_{\text {Blade profile power}}+\underbrace{\eta_2 \sqrt{\sqrt{\eta_3+(v_{j}^{t})^4/4}-(v_{j}^{t})^2/2}}_{\text {Induced power}}\underbrace{+\eta_4(v_{j}^{t})^3}_{\text {Parasite power}}$}, where $v_j^{\text{tip}}$ denotes the tip speed of the rotor blade, and $\eta_1$, $\eta_2$, $\eta_3$, and $\eta_4$ are the constants that depend on the aerodynamic parameters of the UAV. Consequently, the energy consumption of MEC server $j$ to provide computation service for task $\mathcal{K}_{i}^{t}$ can be concluded as:

{\small
\begin{subnumcases}{\label{eq_energyMecComp1}E_{j,i}^t=}
	$$ \gamma_j(f_{j,i}^t)^{2}\mu_{i}^t, \ j=b$$, \label{eq_energyMecComp1A}\\
	$$ \gamma_j(f_{j,i}^t)^{2}\mu_{i}^t + \big( \eta_1\big(1+3(v_{j}^{t})^2/({v_j^{\text{tip}}})^2)+\eta_4(v_{j}^{t})^3\big.\notag \\ \big. +\eta_2 \sqrt{\sqrt{\eta_3+(v_{j}^{t})^4/4}-(v_{j}^{t})^2/2}  \big) \delta, \ j\in \mathcal{U}$$.\label{eq_energyMecComp1B}
\end{subnumcases}
}
%
%
\subsection{Models for QoE, Revenue, and System Utility}
\label{sec_utilityModel}

\subsubsection{QoE of MDs}
\label{sec_MDUtility}

\par The QoE obtained by MD $i$ in time slot $t$ is calculated as the difference between the satisfaction degree of task completion and the costs of task offloading. 

\begin{sequation}
	\begin{aligned}
	\label{eq_MD_utility}
		&U_{i,a}^t=w_i\cdot\underbrace{\log\big(1+\tau_i^t-D_{i,n}^t\big)/\log(1+\tau_i^t)}_{\text{Satisfaction degree}}-(1-w_i)\cdot \\&\big(\overbrace{\mathbb{I}_{(n=0)}\underbrace{E_{i,i}^t/E_{i}^{\max}}_{\text{Cost of energy}}}^{\text{Local computing}}+ \overbrace{\mathbb{I}_{(n=j)} (\underbrace{E_{i,j}^t/ E_{i}^{\max}}_{\text{Cost of energy}}+\underbrace{f_{j,i}^t\cdot p_{j,i}^t /G_i^{\max}}_{\text{Cost of payment}})}^{\text{Edge offloading}}\big), \\
		& \forall i \in \mathcal{I},\ j\in \{b,\mathcal{U}\}, n \in \mathcal{N}, 
	\end{aligned}
\end{sequation}

\noindent where the metrics of satisfaction degree and cost, which have different units, are normalized first and incorporated by the weight parameter $w_i$. Specifically, the satisfaction degree represents MD $i$'s satisfaction level in completing the task, which is commonly modeled as a logarithmic function \cite{Xia2021}. Furthermore, the terms $E_{i,i}^t/E_i^{\max}$ and $E_{i,j}^t/E_i^{\max}$ represent the normalized energy consumption of local computing (i.e., $\mathbb{I}_{(n=0)}=1$) and edge offloading (i.e., $\mathbb{I}_{(n=j)}=1$), respectively, where $E_i^{\max}$ is the energy constraint of MD $i$.  Besides, $f_{j,i}^t\cdot p_{j,i}^t /G_i^{\max}$ represents the normalized payment, where $p_{j,i}^t$ is the unit price of computing resource paid by MD $i$, and $G_i^{\max}$ is the budget of MD $i$.

\subsubsection{Revenue of MEC Servers}
\label{sec_MECUtility}

\par  The utility obtained by MEC server $j\in \{b,\mathcal{U}\}$ from executing task $\mathcal{K}_i^t$ is calculated as the difference between the reward received from MD $i$ and the cost of energy consumption, i.e.,
\begin{sequation}
	\label{eq_MEC_utility}
         \begin{aligned}
	&U_{j,i}^t = w_j\underbrace{f_{j,i}^t p_{j,i}^t/(f_{j}^{\max} p_j^{\max})}_{\text{Reward from MD}}-(1-w_j)\underbrace{E_{j,i}^t/E_j^{\max}}_{\text{Cost of energy}}, \\&\forall i \in \mathcal{I},\ j\in \{b,\mathcal{U}\}.
 \end{aligned}
\end{sequation}

\noindent Similar to Eq. \eqref{eq_MD_utility}, the metrics of reward and cost are normalized first and incorporated by the weight parameter $w_j$. Specifically, $f_{j,i}^t p_{j,i}^t/(f_{j}^{\max} p_j^{\max})$ denotes the normalized reward of  MEC $j$ received by providing computation service for MD $i$, where $p_j^{\max}$ is the maximum price for MEC server $j$'s computing resource. Furthermore, $E_{j,i}^t/E_j^{\max}$ represents the normalized cost of energy consumption of MEC server $j$,  where $E_j^{\max}$ is the energy constraint of MEC server $j$.

\subsubsection{System Utility}
\label{sec_social_welfare}

\par The system utility is calculated by summing the QoE of MDs and the revenue of MEC servers in time slot $t$, i.e.,
\begin{sequation}
	\label{eq_socialwelfare}
	\begin{aligned}
			U^t&=\sum_{i\in \mathcal{I}}\sum_{n\in \mathcal{N}}\zeta_{i}^t\cdot o_{i,n}^t\cdot \big(U_{i,n}^t+U_{n,i}^t\big).
		\end{aligned}
\end{sequation}

%
%
\subsection{Problem Formulation}
\label{sec_problemFormulation}

\par The optimization problem is formulated to maximize the system utility by jointly optimizing the computing resource allocation and pricing strategy $\mathbf{F}^t= \{f_{j,i}^t,p_{j,i}^t\}_{i\in \mathcal{I},j\in \{b,\mathcal{U}\}}$, the computation offloading strategy $\mathbf{O}^t= \{o_{i,n}^t\}_{i\in \mathcal{I}, n \in \mathcal{N}}$,  and the UAV trajectory control $\mathbf{Q}^t= \{\mathbf{q}_j^t\}_{j\in \mathcal{U}}$. Therefore, the problem can be formulated as follows:

{
\begin{subequations}
	\label{eq_problem}
	\begin{alignat}{2}
		\mathbf{P}: \quad &\max_{\mathbf{O},\mathbf{F},\mathbf{Q}}  \ U^t \label{utility}\\
		\text{s.t.} \quad &  o_{i,n}^t\in\{0,1\}, \ \forall i\in \mathcal{I}, \ n\in \mathcal{N}\label{pro_c1}\\
		& \sum_{n\in \mathcal{N}}o_{i,n}^t\leq 1, \ \forall  i\in \mathcal{I}, \  n\in \mathcal{N}\label{pro_c2}\\
		& o_{i,n}^t\cdot D_{i,n}^t\leq \tau_i^t, \ \forall i\in \mathcal{I},  n\in \mathcal{N}, \label{pro_c4}\\
		&\sum_{i\in\mathcal{I}} o_{i,j}^t\cdot f_{j,i}^t \leq  f_j^{\max}, \  \forall j\in \{b, \mathcal{U}\}, \label{pro_c6}\\
		&\sum_{i \in \mathcal{I}} o_{i,j}^t \leq N_{j}^{\text{core}}, \ \forall j \in \{b, \mathcal{U}\}, \label{pro_c7}\\
		&o_{i,j}^t\cdot c_{j,i}^t\cdot f_{j,i}^t\leq G_i^{\max}, \ \forall i\in \mathcal{I}, \ j\in{\{b, \mathcal{U}\}}, \label{pro_c10}\\
		& \eqref{eq_MD_mobility} \sim \eqref{eq_UAV_mob}. \label{pro_c11}				
	\end{alignat}
\end{subequations}
}

\noindent Constraints \eqref{pro_c1} and \eqref{pro_c2} indicate the constraints of offloading strategy. Constraint \eqref{pro_c4} ensures that the delay of completing the task should not exceed the deadline. Constraints \eqref{pro_c6} and \eqref{pro_c7} constrain the computing resources and the number of CPU cores, respectively, for MEC server. Constraint \eqref{pro_c10} guarantees that the price paid by each MD to the MEC server should not exceed its budget. Constraint \eqref{pro_c11} limits the mobility of MDs and UAVs.

\begin{theorem}
	\label{lemma_NP}
	Problem $\mathbf{P}$ is a non-convex MINLP. 
\end{theorem}

\begin{proof}
  A similar proof can refer to \cite{Wang2022AA}.
\end{proof}

%
%
\section{Algorithm}
\label{sec_jointOffloading}

\par To solve problem $\mathbf{P}$, we propose TJCCT, which is comprised of two time-scale optimization methods. Specifically, \textbf{\textit{in the short time scale}}, a price-incentive trading model is constructed based on the bargaining mechanism to facilitate the negotiation between the MDs and the MEC servers for on-demand computing resource allocation and pricing. Furthermore, to deal with the heterogeneity between the computation tasks of MDs and MEC servers, a many-to-one matching is established to stimulate end-edge collaboration for mutual-satisfactory computation offloading.  \textbf{\textit{In the long time scale}}, based on the optimal strategies of computing resource allocation and computation offloading, UAV trajectory is optimized by using convex optimization.

%
%
\subsection{Small Timescale: Computing Resource Allocation and Computation Offloading}
\label{sec_bargain}

\par In each time slot, the strategies of computing resource allocation and computation offloading are decided.

\subsubsection{Computing Resource Allocation}
\label{sec_bargain}

\par As previously discussed, the unit price of computing resources is introduced to capture the task processing costs that should be covered by MDs. This can be viewed as the process of trading wherein each MD that acts as a buyer purchases computing resources from the suitable MEC server that acts as a seller. Consequently, we are motivated to employ the bargaining mechanism to stimulate the negotiation between the MD and MEC server to achieve satisfied trading regarding the on-demand computing resource allocation and pricing.
\par \textbf{\textit{(a) The Optimal Strategy for Computing Resource Allocation.}} Given the computing resource pricing $p_{j,i}^t$, the optimal amount of computing resource allocated by MEC server $j$ to MD $i$ can be obtained by Theorem \ref{lemma_opt_allo}.

\begin{theorem}
	\label{lemma_opt_allo}
	The optimal amount of computing resources that MD $i$ expects to request from the target MEC server $j$ to offload task $\mathcal{K}_i^t$ is determined as $f_{j,i}^{t^*}=2w_i G_i^{\max}/\big(\vartheta(p_{j,i}^t)-\log\big(1+\tau_i^t\big) p_{j,i}^t(1-w_i)\big)$.
\end{theorem}

\begin{proof}
We omit the proof due to the page limitation.
\end{proof}

\par \textbf{\textit{(b) The Optimal Strategy for Computing Resource Pricing.}} Based on the Rubinstein bargaining mechanism \cite{rubinstein1982perfect}, the negotiation between MD $i$ and MEC server $j$ on the price of the computing resource is modeled as the bargaining over the price surplus, which is given by Lemma \ref{lemma_priceRange}.

\begin{lemma}
    \label{lemma_priceRange}
	 The price surplus for a successful negotiation between MD $i$ and MEC server $j$ is given as $\pi_{j,i}^t = \overline{p}_{j,i}^t - \underline{p}_{j,i}^t$, where $\underline{p}_{j,i}^t =(1-w_j)E_{j,i}^tp_j^{\max}f_j^{\max}/(w_jE_j^{\max}f_{j,i}^t)$ and $\overline{p}_{j,i}^t=\big(w_i\log(1+\tau_i^t-D_{i,j}^t)/ ((1-w_i)\log(1+\tau_i^t))-P_i^tl_i^t/(r_{i,j}^t\tau_i^t)\big) (G_i^{\max}/f_{j,i}^t)$ are the lower bound and upper bound for the unit price of the computing resource.
\end{lemma}

\begin{proof}
We omit the proof due to the page limitation.
\end{proof}

\par Based on Lemma \ref{lemma_priceRange}, MD $i$ and MEC server $j$ take turns making offers about how to divide the surplus until reaching a perfect partition, which is given in Lemma \ref{theorem_ne}.

\begin{lemma}
	\label{theorem_ne}
	\par The bargaining model has a unique perfect partition. In the period in which MD $i$ makes a proposal, the optimal partitions are given as:
\vspace{-1pt} 
{\small \begin{subnumcases}{\label{eq_ne}}
	$$\xi_{i,i}^{t^*}=\lambda_{i}^t-\big(1-\lambda_{i}^t\big)\big(1-(\lambda_{i}^t\lambda_{j}^t)^{\lceil \frac{T^b}{2}\rceil}\big)/(1-\lambda_{i}^t\lambda_{j}^t)$$,  \label{eq_optpartionA}\\
	$$\xi_{j,i}^{t^*}=(1-\lambda_{i}^t)\big(2-\lambda_{i}^t\lambda_{j}^t-(\lambda_{i}^t\lambda_{j}^t)^{\lceil \frac{T^b}{2}\rceil}\big)/(1-\lambda_{i}^t\lambda_{j}^t)$$.\label{eq_optAlloPriceB}
\end{subnumcases}}

\noindent where $\lambda_{i}^t=1-l_i^t/(r_{i,j}^t\tau_i^t)$ and $\lambda_{j}^t=1-\mu_i^t/(f_{j,i}^t\tau_i^t)$ denote the discount factors of MD $i$ and MEC server $j$, which 
are employed to evaluate the patience with the negotiation delay. 
\end{lemma}

\begin{proof}
	\label{eq_optPartition}
	A similar proof can refer to \cite{binmore1986nash}.
\end{proof}

\begin{theorem}
	\label{lemma_pricing}	
	\par The optimal pricing of computation resource $p_{j, i}^{t^*}$ is obtained as: in the period when MD $i$ makes an offer,  $p_{j,i}^{t^*}=\overline{p}_{j,i}^t-\pi_{j,i}^t\cdot {\xi_{i,i}^{t^*}}$, and in the period when MEC server $j$ makes an offer, $	p_{j,i}^{t^*}=\overline{p}_{j,i}^t-\pi_{j,i}^t \cdot \xi_{i,j}^{t^*}$.
\end{theorem}

\begin{proof}
\label{eq_optPrice}
    We omit the proof due to the page limitation.
\end{proof}

\begin{corollary}
\label{cor_deal}

\par It can be concluded from Theorems \ref{lemma_opt_allo} and \ref{lemma_pricing} that a trading consensus can be reached on the computing resource allocation and pricing:
\begin{sequation}
	\label{eq_optAllo}
	f_{j,i}^{t^*}=2w_iG_i^{\max}/\big(\vartheta(p_{j,i}^{t^*})-\log\big(1+\tau_i^t) \cdot p_{j,i}^{t^*} \cdot (1-w_i)\big),
\end{sequation}

{\small \begin{subnumcases}{\label{eq_optAlloPrice}p_{j,i}^{t^*}=}
	$$\overline{p}_{j,i}^t-\Delta p_{j,i}^t\cdot \xi_{i,i}^{t^*}$$,  \label{eq_optAlloPriceA}\\
	$$\overline{p}_{j,i}^t-\Delta p_{j,i}^t\cdot  \xi_{i,j}^{t^*}$$.\label{eq_optAlloPriceB}
\end{subnumcases}}
\end{corollary}

\par The trading contract between MD $i$ and MEC server $j$ is presented in Definition \ref{def_pricingRule} as follows.

\begin{definition}
	\label{def_pricingRule}
	Trading contract. The trading amount and price are determined based on the following terms.
	\begin{itemize}
		\item If $U_{i,j}^t>0$, $U_{j,i}^t>0$, a consensus in Eqs. \eqref{eq_optAllo} and \eqref{eq_optAlloPrice} is reached.
		\item If $U_{i,j}^t>0$, $U_{j,i}^t<0$, MD $i$ makes an offer of the computing resource pricing $p_{j,i}^{t^*}$ based on Eq. \eqref{eq_optAlloPriceA}.
		
		\item If $U_{i,j}^t<0$, $U_{j,i}^t>0$, MEC server $j$ makes an offer of the computing resource pricing $p_{j,i}^{t^*}$ based on Eq. \eqref{eq_optAlloPriceB}.
		
		\item If $U_{i,j}^t<0$, $U_{j,i}^t<0$,  either MD $i$ or MEC server $j$ can make an offer of the computing resource pricing $p_{j,i}^{t^*}$ based on Eq. \eqref{eq_optAlloPrice}.
	\end{itemize}
\end{definition}

\par According to Corollary \ref{cor_deal} and Definition \ref{def_pricingRule}, the algorithm for computing resource allocation and pricing is described in Algorithm \ref{algo_allo_price}. Specifically, in each iteration, MD $i$ and MEC server $j$ negotiate the optimal strategy of pricing based on the trading contract (line 6). Then update the optimal strategy of computing resource allocation (line 7). The above steps are iterated until a consensus is reached.

\begin{algorithm}[]	
	\label{algo_allo_price}	
	\SetAlgoLined
	\KwIn{MD $i$, MEC server $j$}
	\KwOut{The optimal resource allocation and pricing strategy $(f_{j,i}^{t^*}, p_{j,i}^{t^*})$ in time slot $t$}
	\textbf{ Initialization:} 
	$U_{i,j}^t= 0$; $U_{j,i}^t= 0$; $\iota^{\max}= 100$\;
	Set the optimal resource allocation as $f_{j,i}^{t^*} = f_j^{\text{avl}}$\;
	\While{$\iota \leq \iota^{\max}$}
	{
		Update $p_{j,i}^{t^*}$ based on Eq. \eqref{eq_optAlloPrice}\;
		Calculate $U_{i,j}^t$, $U_{j,i}^t$ based on Eqs. \eqref{eq_MD_utility} and \eqref{eq_MEC_utility}\;
		Perform the trading contract based on Definition \ref{def_pricingRule}\;
		Update $f_{j,i}^{t^*}$ based on Eq. \eqref{eq_optAllo}\;
  	$\iota = \iota +1$\;
	}
	\Return{$(f_{j,i}^{t^*}, p_{j,i}^{t^*})$}\;
	\caption{Computing Resource Allocation.}
\end{algorithm}	

%
%
\subsubsection{Computation Offloading}
\label{sec_matching}

\par  Matching mechanism offers an efficient tool to construct the mutual-beneficial relationship between two sets of entities with heterogeneous preferences. This motivates us to construct the matching between the computation tasks of MDs and MEC servers to alleviate the demand-supply heterogeneity. By doing so, the MDs and MEC servers can achieve mutual-beneficial computation offloading results of satisfied QoE and high computing resource utilization. Denote the set of computation tasks that have not begun execution in time slot $t$ as $\mathcal{K}_{\text{req}}^t=\{\mathcal{K}_i^{\mathbf{t}}|i\in \mathcal{I}, \mathbf{t}\in \mathcal{T}\}$, where $\mathbf{t}$ is the generation time of the computation task. The offloading strategy for these computation tasks in each time slot is decided using a many-to-one matching mechanism, which is defined as follows.

\begin{definition}
	The current matching is defined as a triplet of $(\mathcal{M}^t,\mathcal{L}^t, \Pi^t)$: 
	\begin{itemize} 
		\item $\mathcal{M}^t=\big(\mathcal{K}_{\text{req}}^{t}, \{b,\mathcal{U}\}\big)$ denotes the computation tasks and MEC servers.
		\item $\mathcal{L}^t=\big(\mathcal{L}_{\mathcal{K}_i^{\mathbf{t}}}^t,\mathcal{L}_{j}^t\big)$ consists of the preference lists of the computation tasks and MEC servers. Each computation task $\mathcal{K}_i^{\mathbf{t}}\in\mathcal{K}_{\text{req}}^{t}$ has a descending ordered preferences over the MEC servers, i.e., $\mathcal{L}_{\mathcal{K}_i^{\mathbf{t}}}^t=\{j|j\in\{b,\mathcal{U}\}, j\succ_{\mathcal{K}_i^{\mathbf{t}}}{j^\prime}\}$, where $\succ_{\mathcal{K}_i^{\mathbf{t}}}$ is the preference of computation task $\mathcal{K}_i$ towards the MEC servers. Moreover, each MEC server $j \in \{b,\mathcal{U}\}$ has a descending ordered preference list over the tasks, i.e., $\mathcal{L}_j^t=\{\mathcal{K}_i^{\mathbf{t}}\in \mathcal{K}_{\text{req}}^{t}, \mathcal{K}_i^{\mathbf{t}} \succ_{j} {\mathcal{K}_i^{\mathbf{t}}}^{\prime}\}$. 
	
		\item $\Pi^t\subseteq \mathcal{K}_{\text{req}}^{t} \times \{b,\mathcal{U}\}$ is the many-to-one matching between the tasks and MEC servers. Each task $\mathcal{K}_i^{\mathbf{t}}\in \mathcal{K}_{\text{req}}^{t}$ can be matched with at most one MEC server, i.e., $\Pi_{\mathcal{K}_i^{\mathbf{t}}}^t\in \{b,\mathcal{U}\}$, while each MEC server $j\in\{b,\mathcal{U}\}$ can be matched with multiple tasks, i.e., $\Pi_j^t\subseteq \mathcal{K}_{\text{req}}^{t}$.
	\end{itemize}
\end{definition}

\begin{algorithm}[]	
	\label{algo_matching}	
	\SetAlgoLined
	\KwIn{Tasks $\mathcal{K}_{\text{req}}^t=\{\mathcal{K}_i^{\mathbf{t}}|i\in \mathcal{I}, t\in \mathcal{T}\}$, and MEC servers  $\{b,\mathcal{U}\}$}
	\KwOut{The optimal matching list $\Pi^{t^*}$, offloading $\mathbf{O}^{t^*}$, and computing resource allocation $\mathbf{F}^{t^*}$}
	\textbf{ Initialization:} 
	$\mathcal{K}_{\text{rej}}^t=\mathcal{K}_{\text{req}}^{t}$, $\Pi^{t^*}=\emptyset$\;
	\For{$\mathcal{K}_i^{\mathbf{t}}\in \mathcal{K}_{\text{req}}^{t}$}
	{
		\For {$j\in \{b,\mathcal{U}\}$}
		{
			Call Algorithm \ref{algo_allo_price} to obtain $\big(f_{j,i}^{\mathbf{t}^*},p_{j,i}^{\mathbf{t}^*}\big)$\;
			Calculate $V_{\mathcal{K}_i^{\mathbf{t}},j}^t=U_{i,j}^{\mathbf{t}}$,$\ V_{j,\mathcal{K}_i^{\mathbf{t}}}=U_{j,i}^{\mathbf{t}}$\;
			$V_{\mathcal{K}_i^{\mathbf{t}},j}^t>V_{\mathcal{K}_i^{\mathbf{t}},j^{\prime}}^t \Leftrightarrow j\succ_{\mathcal{K}_i^{\mathbf{t}}} j^{\prime}, \ \mathcal{L}_{\mathcal{K}_i^{\mathbf{t}}}=\{j,j^{\prime}\}$\;
			$V_{j,\mathcal{K}}>V_{j,\mathcal{K}^{\prime}} \Leftrightarrow \mathcal{K}\succ_j \mathcal{K}^{\prime}, \ \mathcal{L}_j^t=\{\mathcal{K},\mathcal{K}^{\prime}\}$ \;
		}
	}
	\While{\rm{There exists} $\mathcal{K}_i^{\mathbf{t}}\in \mathcal{K}_{\text{rej}}^t$: $\mathcal{L}_{\mathcal{K}_i^{\mathbf{t}}}^t \neq \emptyset \ \&\& \ \mathcal{K}_i^{\mathbf{t}} \notin \mathcal{L}_j^t $}
	{
		\For{$\mathcal{K}_i^{\mathbf{t}}\in \mathcal{K}_{\text{rej}}^t$}
		{$\Pi_{\mathcal{K}_i^{\mathbf{t}}}^t=\Pi_{\mathcal{K}_i^{\mathbf{t}}}^t\cup j^{\prime}$, $\ j^{\prime} =  \mathcal{L}_{\mathcal{K}_i^{\mathbf{t}}}^t[1]$   \; 
			\If{$V_{\mathcal{K}_i^{\mathbf{t}},j^{\prime}}^t>0$}
			{
				 $\Pi_{j^{\prime}}^t=\Pi_{j^{\prime}}^t\cup \mathcal{K}_i^{\mathbf{t}}$
			}
		}
		\For{$j\in \{b, \mathcal{U}\}$ \rm{that receives new requests}}
		{
			 $|\Pi_j^t|\leq N_j \leq N_{j}^{\text{idl}}\ $,
			 $\sum_{\mathcal{K}_i^{\mathbf{t}} \in \Phi(j)} f_{j,i}^{\mathbf{t}^*} \leq f_j^{t,\text{avl}}$\;
			  $\Pi_j^t=\Pi_j^t\ \backslash \ \mathcal{D}_j^t$,$\ \mathcal{K}_{\text{rej}}^t=\mathcal{K}_{\text{rej}}^t\cup \mathcal{D}_j^t$ \;	
			{
				\For {$\mathcal{K}_i^{\mathbf{t}} \in \mathcal{D}_j^t$}
				{$\mathcal{L}_{\mathcal{K}_i^{\mathbf{t}}}^t=\mathcal{L}_{\mathcal{K}_i^{\mathbf{t}}}^t\ \backslash \ \{j\}$,$\ \Pi_{\mathcal{K}_i^{\mathbf{t}}}^t=\Pi_{\mathcal{K}_i^{\mathbf{t}}}^t\ \backslash \ \{j\}$; 
				}
			}	
		}	 	
	}	 	 	
	\Return {$\Pi^{t^*}=\Pi^t$, $\mathbf{O}^{t^*}=\{o_{i,j}^{t}| j=\Pi_{\mathcal{K}_i^{\mathbf{t}}}^t, \mathcal{K}_i^{\mathbf{t}}\in \mathcal{K}_{\text{req}}^{t}\}$, $\mathbf{F}^{t^*}=\{\big(f_{j,i}^{\mathbf{t}^*}, p_{j,i}^{\mathbf{t}^*}\big)|j=\Pi_{\mathcal{K}_i^{\mathbf{t}}}^t, \mathcal{K}_i^{\mathbf{t}}\in \mathcal{K}_{\text{req}}^{t}\}$}\;
\caption{Computation Offloading.}
\end{algorithm}

\par The main steps of the matching process are presented in Algorithm \ref{algo_matching}, which are described as follows.

\textit{\textbf{Preference List Construction.}} For each task $\mathcal{K}_i^{\mathbf{t}}\in \mathcal{K}_{\text{req}}^{t}$ and MEC server $j\in \{b,\mathcal{U}\}$, the preference lists are constructed based on the following steps: \textbf{\textit{i)}} predict the optimal resource allocation and pricing by calling Algorithm \ref{algo_allo_price} (line 4), \textbf{\textit{ii)}} calculate the values of preference for each task (MEC server) on MEC servers (tasks) (line 5), \textbf{\textit{iii)}} construct the preference list for each task and MEC server by ranking the preference values in descending order (lines 6 and 7).

\par \textit{\textbf{Matching Construction.}} The matching process is implemented according to the following steps: \textbf{\textit{i)}} for each computation task $\mathcal{K}_i^{\mathbf{t}}\in \mathcal{K}_{\text{rej}}^t$, select the most preferred MEC server $j^{\prime}$ and add it to the matching list temporarily (line 10), \textbf{\textit{ii)}} if the computation task prefers MEC server $j^{\prime}$, add the computation task to the matching list of $j^{\prime}$ temporarily (lines 11 and 12), \textbf{\textit{iii)}} for each MEC server that receives new requests, update the matching list by remaining the top\textendash $N_j$ most preferred computation tasks and removing the less preferred computation tasks (lines 13 to 14) to guarantee that current number of tasks and the allocated computing resources should not exceed the number of idle CPU cores $N_j^{t,\text{idl}}$ and the available computing resources $f_j^{t,\text{avl}}$ of the MEC server, \textbf{\textit{iv)}} add the deleted computation tasks into the rejected set, \textbf{\textit{v)}} update the preference list and matching list for the deleted computation tasks (lines 16 and 17). Repeat the above steps until all computation tasks have been matched with an MEC server, or the unmatched computation tasks have been rejected by all MEC servers.

\vspace{-3pt}

\subsection {Large Timescale: UAV Trajectory Control}
\label{sec_UAV Trajectory Control}

\par In each time epoch, the UAV trajectory is optimized by applying the convex approximation method. Specifically, with the optimized strategies of computing resource allocation and computation offloading, and eliminating the unrelated terms in the objective function and constraints, the problem of UAV trajectory optimization can be given as:
\begin{subequations}
	\label{eq_problem1}
	\begin{alignat}{2}
		\mathbf{P_t}: \ &\max_{\mathbf{Q}^{t^\prime}} U^t=\max_{\mathbf{Q}^{t^\prime}}\sum_{i\in \mathcal{I}}\sum_{j\in \mathcal{U}}\zeta_{i}^to_{i,j}^t \big(U_{i,j}^t+U_{j,i}^t\big) \notag\\
		&\eqref{eq_UAV_mob_ini_fin} \sim \eqref{eq_UAV_mob_dis_each1},\notag
	\end{alignat}
\end{subequations}

\noindent where $\mathbf{Q}^{t^\prime}$ denotes the positions of UAVs in the next time epoch $t^\prime=\big(\lceil t/\Delta \rceil+1\big)\Delta$.

\begin{lemma}
	\label{lem_appro_p1}
	Problem $\mathbf{P_t}$ can be approximately converted as:
	{\small
	\begin{subequations}
		\label{eq_problem11}
		\begin{alignat}{2}
			\overline{\mathbf{P}}_\mathbf{t}: \ &\max_{\mathbf{Q}^{t^\prime}}  \ \sum_{i\in \mathcal{I}}\sum_{j\in \mathcal{U}}\zeta_{i}^t  o_{i,j}^t  \big(\vartheta_0\log\big(1+\big(\tau_i^t-\frac{l_{i}^{t}}{\overline{r}_{i,j}^{t^\prime}}-\vartheta_1\big)\big)\big.\\
			&\eqref{eq_UAV_mob_ini_fin} \sim \eqref{eq_UAV_mob_dis_each1}.\notag
		\end{alignat}
	\end{subequations}
    }
	\noindent where $\vartheta_0 = w_i/(1+\tau_i^t)$, $\vartheta_1=\mu_i^t/f_{j,i}^t$, and $\overline{r}_{i,j}^{t^\prime} = B_{i,j}^t\log_2 \big(1+P_{i}^t\bar{g}_{i,j}^t/N_0 \big(\|{\mathbf{q}}_j^{t^\prime}-\mathbf{q}_i^t\|^2 + H^2\big)^{\beta_{A}/2}\big)$.
\end{lemma}

\begin{proof}
We omit the proof due to the page limitation.
\vspace{-1pt}
\end{proof}

\par  However, Problem $\overline{\mathbf{P}}_\mathbf{t}$ is a non-convex optimization problem due to the non-concavity of the objective function and the non-convexity of Constraint \eqref{eq_UAV_mob_dis_each}. Therefore, it will be transformed into a convex problem by the following steps.

\par \textbf{First}, since the objective function of $\overline{\mathbf{P}}_\mathbf{t}$ is non-convex with respect to $\overline{r}_{i,j}^{t^\prime}$, the auxiliary variables $\tilde{r}_{i,j}^{t^\prime}$ is first introduced such that $\tilde{r}_{i,j}^{t^\prime} \leq \overline{r}_{i,j}^{t^\prime}$, where the RHS is lower bounded by a concave function as given in Lemma \ref{lemma_r_convex}.

\begin{lemma}
	\label{lemma_r_convex}
	Given the local point $\hat{\mathbf{q}}_j^{s}$ at the $s$-th iteration, $\overline{r}_{i,j}^t$ is lower bounded by:
 
	\begin{sequation}
		\label{eq_r_firstTaylor}
		\begin{aligned}
			&\overline{r}_{i,j}^t \geq  B_{i,j}\log_2\big(1+P_{i}^t\bar{g}_{i,j}^t/N_0(H^2+\|\hat{\mathbf{q}}_j^{s}-\mathbf{q}_i^t\|^2)^{\beta/2}\big)-B_{i,j}\beta\\
			&\big(\|{\mathbf{q}}_j^{t^\prime}-\mathbf{q}_i^t\|^2-\|\hat{\mathbf{q}}_j^{s}-\mathbf{q}_i^t\|^2)/(2\ln2 (H^2+\|\hat{\mathbf{q}}_j^{s}-\mathbf{q}_i^t\|^2)\big)= \tilde{\tilde {r}}_{i,j}^{s}.
		\end{aligned}
	\end{sequation}
\end{lemma}

\begin{proof}
A similar proof can refer to ~\cite{Yang2022}.
\end{proof}

\par \textbf{Second}, to deal with the non-convexity of the UAV propulsion energy $E_{j}^{p}$, we introduce auxiliary variable $\phi$ such that
\begin{sequation}
	\label{eq_aux_vel}
	\phi^2 \geq \sqrt{\eta_3+(v_j^t)^4/4}-(v_j^t)^2/2\Longrightarrow \eta_3/(v_j^t)^2 \leq \phi^2+(v_j^t)^2,
\end{sequation}

\noindent where $v_j^t=\|\mathbf{q}_j^{t^\prime}-\mathbf{q}_j^t\|/(\delta\Delta)$. For the convex RHS of \eqref{eq_aux_vel}, a global concave lower bound can be obtained at the local point $\phi^s$ by using the first-order Taylor expansion, i.e.,
\begin{sequation}
	\label{eq_aux_vel1}
	\begin{aligned}
		 \phi^2+(v_j^t)^2\geq&(\phi^s)^2+2 \phi^s(\phi-\phi^s) +\|\hat{\mathbf{q}}_j^{s}-\mathbf{q}_i^t\|/\Delta^2\\&+2(\hat{\mathbf{q}}_j^{s}-\mathbf{q}_i^t)^T(\hat{\mathbf{q}}_u^{t^\prime}-\mathbf{q}_i^t)/\Delta^2=\tilde{\phi}^s\
	\end{aligned}
\end{sequation}

\par Based on Lemmas \ref{lem_appro_p1} and \ref{lemma_r_convex}, by introducing the auxiliary variables $\overline{r}_{i,j}^t$,  $\tilde{\tilde{r}}_{i,u}^s$, $\phi$, $\tilde{\phi}^s$, Problem $\overline{\mathbf{P}}_\mathbf{t}$ can be transformed as:
{\small
\begin{subequations}
	\label{eq_problemSp1_fin}
	\begin{alignat}{2}
		\overline{\mathbf{P}}_\mathbf{t1}: &\max_{\mathbf{Q}{t^\prime},}  \ \sum_{i\in \mathcal{I}}\sum_{j\in \mathcal{U}}\zeta_{i}^t  o_{i,j}^t  \Big(\vartheta_0\log(1+\tau_i^t-l_{i}^{t}/\overline{r}_{i,j}^{t^\prime}-\vartheta_1)\Big.\notag\\
		&\Big.-\vartheta_2l_i^t/\overline{r}_{i,j}^t-\vartheta_3\big(\eta_1(1+3 (v_{j}^{t})^2/{v_j^{\text{tip}}}^2)+\eta_2 \phi+\eta_4 (v_{j}^{t})^3\big)\delta\Big)\\
		\text{s.t.} \ &\overline{r}_{i,j}^t \leq \tilde{\tilde {r}}_{i,j}^s,\ \forall i \in \mathcal{I}, \  j \in \mathcal{U}, \ t \in \mathcal{T}, \label{eq_Sp1_fin_c1} \\	
		& \eta_3/(v_j^t)^2 \leq \tilde{\phi}^s, \ \forall i \in \mathcal{I}, \  j \in \mathcal{U}, \ t \in \mathcal{T}, \label{eq_Sp1_fin_c2}\\
		 &\eqref{eq_UAV_mob_ini_fin} \sim \eqref{eq_UAV_mob_dis_each}, \notag
	\end{alignat}
\end{subequations}
}

\noindent where Problem $\overline{\mathbf{P}}_\mathbf{t1}$ is convex since the objective function is concave and the feasible region is convex, which can be easily solved by optimization tools such as CVX. The solution of UAV trajectory control is summarized in Algorithm \ref{alg_Traj}. First, in the $s$-th iteration, the lower bounds $\tilde{\tilde {r}}_{i,j}^{s}$ and $\tilde{\phi}^s$ are calculated (line 3). Then, the optimal trajectory $\mathbf{Q}^{s^*}$ of Problem $\overline{\mathbf{P}}_\mathbf{t1}$ is obtained as the local point for the next iteration $\mathbf{Q}^{s+1}$ (lines 4 and 5). The iteration ends when the difference of the objective value falls below a given threshold $\epsilon$.

\begin{algorithm}[]	
	\label{alg_Traj}	
	\SetAlgoLined
	\KwIn{UAV location $\mathbf{Q}^{t}$, optimal offloading strategy $\mathbf{F}^{t^*}$ and optimal resource allocation strategy $\mathbf{O}^{t^*}$ in time slot $t$}
	\KwOut{UAV trajectory in the next time epoch $\mathbf{Q}^{t^{\prime*}}$}
	\textbf{ Initialization:} 
	$\epsilon$, $s=0$, ${\hat{\mathbf{q}}}_j^{s}=\mathbf{q}_j^t$, $U^{s}=0$\;
	\Repeat{$|U^{s}-U^{s-1}|\leq \epsilon$}
	{
		Calculate $\tilde{\tilde {r}}_{i,j}^{s}$ and $\tilde{\phi}^s$ based on Eqs. \eqref{eq_r_firstTaylor} and \eqref{eq_problemSp1_fin}\;
		Solve Problem $\overline{\mathbf{P}}_\mathbf{t1}$ to obtain the optimal trajectory $\mathbf{Q}^{s^*}$ and objective value $U^{s^*}$\;
		Update $\mathbf{Q}^{s+1}=\mathbf{Q}^{s^*}$\;
	 	Update s = s + 1\;
	}	 	 	
	\Return {$\mathbf{Q}^{t^{\prime*}}$}\;
	\caption{UAV Trajectory Control.}
        \vspace{-2pt}
	\end{algorithm}



\begin{algorithm}[!hbt]	
	\label{algo_TJCCT}	
	\SetAlgoLined
	\KwIn{$\mathcal{I}$, $\{b, \mathcal{U}\}$, $\mathcal{T}$}
	\KwOut{$U$}
	\textbf{Initialization:} $t=0$, $U=0$\;
	\While{$t\leq T$}
	{		
		Each MD processes the task locally if $U_{i,0}^t>0$ \;
		Call Algorithm \ref{algo_matching} to obtain $\Pi^{t^*}$,  $\mathbf{O}^{t^*}$, and $\mathbf{F}^{t^*}$\;
			\For{ $\Pi_j^t \in \Pi^{t^*}$}{
				Perform resource allocation and charging\;
				Update the task processing state and available resources of MEC servers\;
                Calculate the current system utility $U^t$\;
			}				
		Calculate the total system utility $U=U+U^t$\;
		\If{$t \ \% \ \Delta == 0$}{
            Call Algorithm \ref{alg_Traj} to obtain $\mathbf{Q}^{t^{\prime*}}$\;
			Update the mobility of MDs and UAVs\;
		}
		 $t = t+\delta$\;				
	}
	\Return $U$\;
	\caption{TJCCT}
\end{algorithm}

%
%

\section{Main Steps of TJCCT and Performance Analysis}
\label{sec_stability}
\par In this section, the main steps of TJCCT are given in Algorithm \ref{algo_TJCCT}. Then, the stability, optimality, and computational complexity of TJCCT are proved as follows.

\par \textit{\textbf{Stability.}} The stability of TJCCT depends on the decision of computation offloading, which relies on the result of matching. Based on Theorem \ref{theo_stable}, the stability of TJCCT is proved.

\begin{theorem}
\label{theo_stable}
The result of matching $\Pi^{t^*}$ is stable.
\end{theorem}

\begin{proof}
A similar proof can refer to \cite{Wang2022a}.
\vspace{-1pt}
\end{proof}

\par \textit{\textbf{Optimality.}} For computing resource allocation, the optimality can be easily obtained based on the theoretical result in Corollary \ref{cor_deal}. For computation offloading, the weak-Pareto optimality is proved in Theorem \ref{the_weak_pare}. For UAV trajectory control, the optimality is proved in Theorem \ref{the_opt_convex}. As a result, the optimality of TJCCT is proved.

\begin{theorem}
\label {the_weak_pare}
The result of computation offloading $\mathbf{O}^{t^*}$ is weak Pareto optimal.
\end{theorem}

\begin{proof}
A similar proof can refer to \cite{zhou2017social}.
\vspace{-1pt}
\end{proof}

\begin{theorem}
\label {the_opt_convex}
 Problem $\overline{\mathbf{P}}_\mathbf{t}$ does not change the optimality of Problem $\mathbf{P_t}$.
 \vspace{-1pt}
\end{theorem}

\begin{proof}
We omit the proof due to the page limitation. 
\vspace{-1pt}
\end{proof}

\par \textit{\textbf{Computation Complexity.}} The computational complexity of TJCCT is given as Theorem \ref{theo_complexity}.

\begin{theorem}
\label{theo_complexity}
TJCCT has a polynomial worst-case complexity in each time slot, i.e., $\mathcal{O}\big(\iota^{\max}\big(|\mathcal{U}|+1\big)\big(2|\mathcal{K}_{\text{req}}^{t}|+ \min\{|\mathcal{U}|+1, |\mathcal{K}_{\text{req}}^{t}|\}\big)+({IU})^{3.5}\big.$ $\big.\log_2(1/\epsilon)\big)$, where $|\mathcal{K}_{\text{req}}^{t}|$ and $|\mathcal{U}|$ are the numbers of undecided tasks and UAVs, respectively.
\vspace{-1pt}
\end{theorem}

\begin{proof}
	We omit the proof due to the page limitation.
\end{proof}

\section{Simulation Results and Analysis}
\label{sec_simulation}

\par In this section, simulation results are presented to validate the effectiveness of the proposed approach.

\subsection{Simulation Setup}
\label{simulation_set_up}

\par \textbf{Scenarios.} We consider a UAV-assisted MEC system where 1 terrestrial MEC server and 2 aerial MEC servers are deployed to provide service for 20 MDs in a $500\times 500 \ \text{m}^2$ rectangular area. The time horizon is set as $50$ s, and is divided into $T = 500$ time slots with equal length $\delta=100$ ms, and 10 time slots are grouped into a time epoch, i.e., $\Delta=1$ s. 

\par \textbf{Parameters.} For MDs, the related parameters are set as : computing resources $f_i^{\max}=[0.5,1]$ GHz and transmit power $p_i^t=[10,30]$ dBm. For tasks, the related parameters are set as follows: task size $l_i^t\in [1,5]$ Mb \cite{Yu2020}, required computing resources $\mu_i^t=[500, 1500]$ cycles/bit, and deadline $\tau_i^{\max}=[0.5, 5]$ s. For the terrestrial MEC server $j=b$, the related parameters are set as follows: the fixed location $[250,250]$, the computation resources $f_j^{\max} \in [20,40]$ GHz, and the number of CPU cores $N_j^{\text{core}}\in[4,8]$. For aerial MEC server $j\in \mathcal{U}$, the related parameters are set as follows: fixed altitude $H$ = 100 m, maximum velocity $v_\text{U}^{\max}=25$ m/s, initial positions and destinations $(\textbf{q}_{1}^{I},\textbf{q}_{1}^{F})=([0,0],[500,0])$, $(\textbf{q}_{1}^{I},\textbf{q}_{1}^{F})=([500,0],[0,0])$, computation resources $f_j^{\max} \in [10,20]$ GHz \cite{Zhou2022} and the number of CPU cores $N_j^{\text{core}}\in[2,4]$. The setting of communication-related parameters follows the 3GPP specification \cite{3GPPTR389012020}.

\par \textbf{Benchmarks.} This work evaluates the proposed TJCCT in comparison with the following benchmark algorithms. \textit{Local-only strategy (LS)}: all MDs process the tasks locally. \textit{Greedy strategy (GS)}: each MD greedily acquires the optimal strategies for task offloading and computing resource allocation to maximize its utility. \textit{Nearest strategy (NS)}: each MD offloads the tasks to the nearest MEC server to which it is attached. \textit{Game theoretic strategy (GTS)}: this algorithm is extended from \cite{Wang2020} and adjusted accordingly to suit for this work. \textit{Cooperative scheme (CS)} \cite{Wang2022a}: this UAV-UGV cooperative computation offloading scheme is extended to be suited to the problem in this work. Specifically, the strategies of computation offloading and UAV trajectory control are based on the one-to-one matching mechanism and the segment-constrained method, respectively. Note that the strategy of trajectory control for LS, GS, NS, and GTS is determined based on the method proposed in this work.

\textbf{Performance Indicators.} To evaluate the overall performance of the proposed method, we adopt the total utility of the system, i.e., $ \sum_{t\in\mathcal{T}} U^t$, the aggregate QoE of MDs $\sum_{t\in\mathcal{T}}\sum_{i \in \mathcal{I}}\sum_{n \in \mathcal{N}}U_{i,n}^t$, and the total revenue of MEC servers, i.e., $\sum_{t\in\mathcal{T}}\sum_{j \in \{b,\mathcal{U}\}}\sum_{i \in \mathcal{I}}U_{j,i}^t$. 


%
%

\subsection{Numerical Results}

\par \textbf{\textit{Effect of Time.}} Figs. \ref{fig_sw}(a), \ref{fig_sw}(b), and \ref{fig_sw}(c) compare the total utility of the system, the aggregate QoE of MDs, and the total revenue of MEC servers, respectively with respect to time. It can be observed that TJCCT outperforms the other algorithms in terms of the total utility of the system, the aggregate QoE of MDs, and the total revenue of MEC servers, with significant performance advantages as time elapses. This can be attributed to several reasons. First, price-incentive resource trading stimulates the MDs and MEC servers to negotiate the on-demand computing resource allocation. Additionally, the matching mechanism employed by TJCCT enables mutually satisfactory computation offloading between MDs and terrestrial-aerial MEC servers based on the available computing resources of MEC servers and the QoE requirements of MDs. Moreover, the UAVs dynamically adjust their trajectories to provide real-time offloading services for MDs by using the trajectory control method of TJCCT. In conclusion, this set of simulation results demonstrates the superiority of TJCCT among the six algorithms, especially in bringing long-time benefits for both MDs and MEC servers.

\begin{figure*}[!hbt] 
	\centering
		\setlength{\abovecaptionskip}{1pt}%
		\setlength{\belowcaptionskip}{1pt}%
	\subfigure[The total utility of the system]
	{
		\begin{minipage}[t]{0.31\linewidth}
			\raggedleft
			\includegraphics[scale = 0.22]{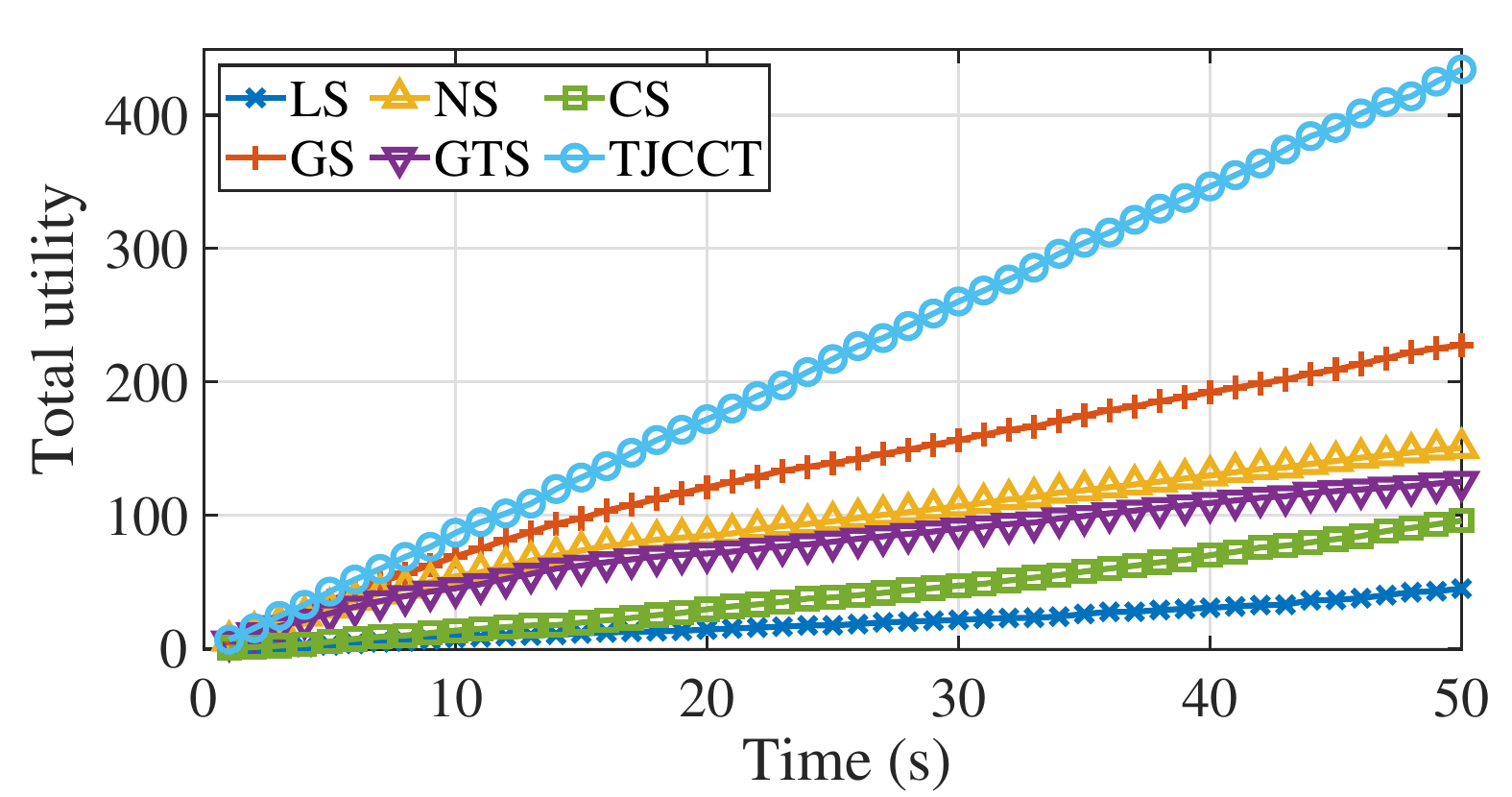}
		\end{minipage}
  \vspace{-1pt}
	}
	\subfigure[The aggregate QoE of MDs]
	{
		\begin{minipage}[t]{0.31\linewidth}
			\centering
			\includegraphics[scale = 0.22]{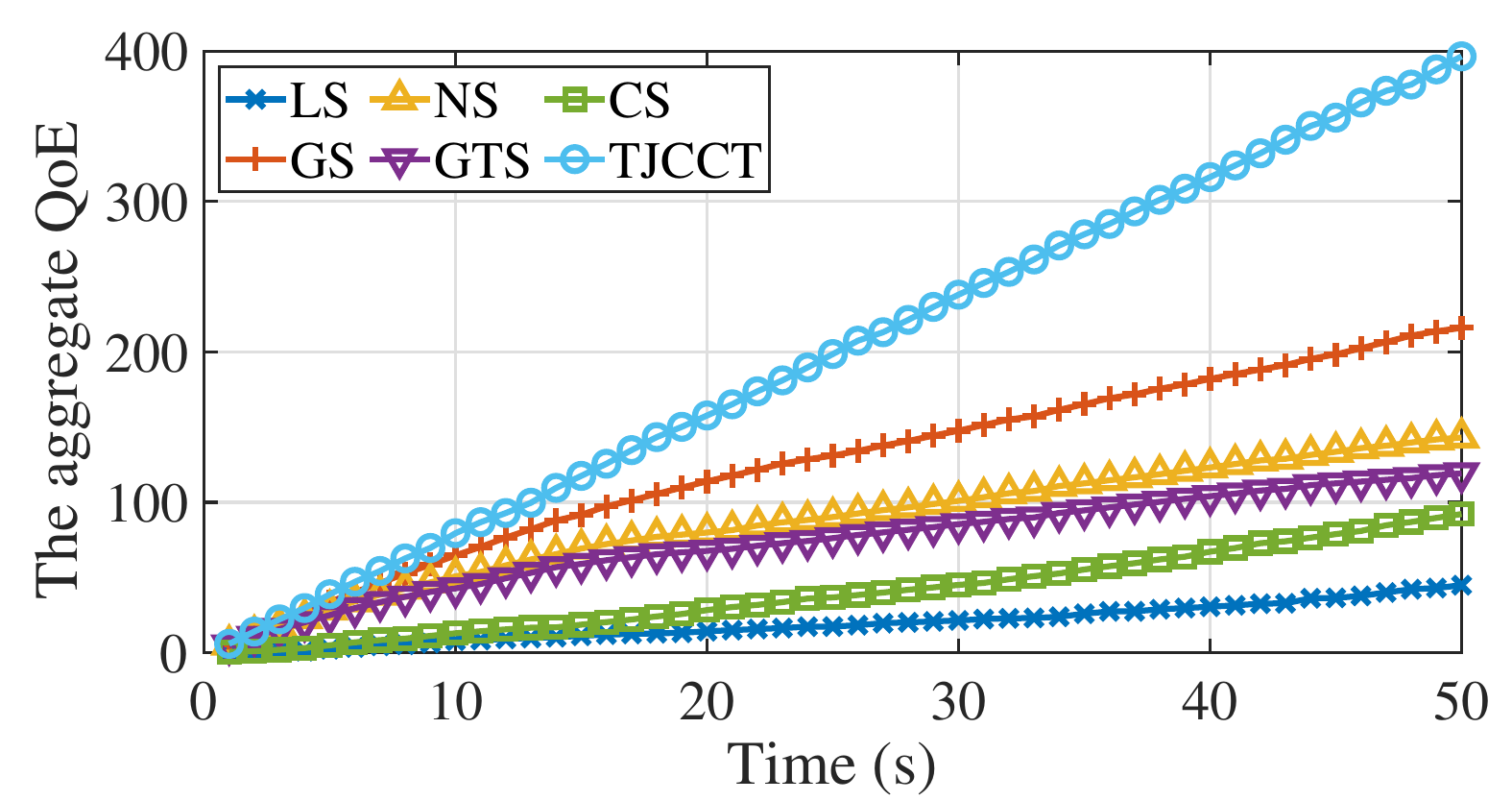}	
		\end{minipage}
  \vspace{-1pt}
	}
	\subfigure[The total revenue of MEC servers]
	{
		\begin{minipage}[t]{0.31\linewidth}
			\raggedright
			\includegraphics[scale = 0.22]{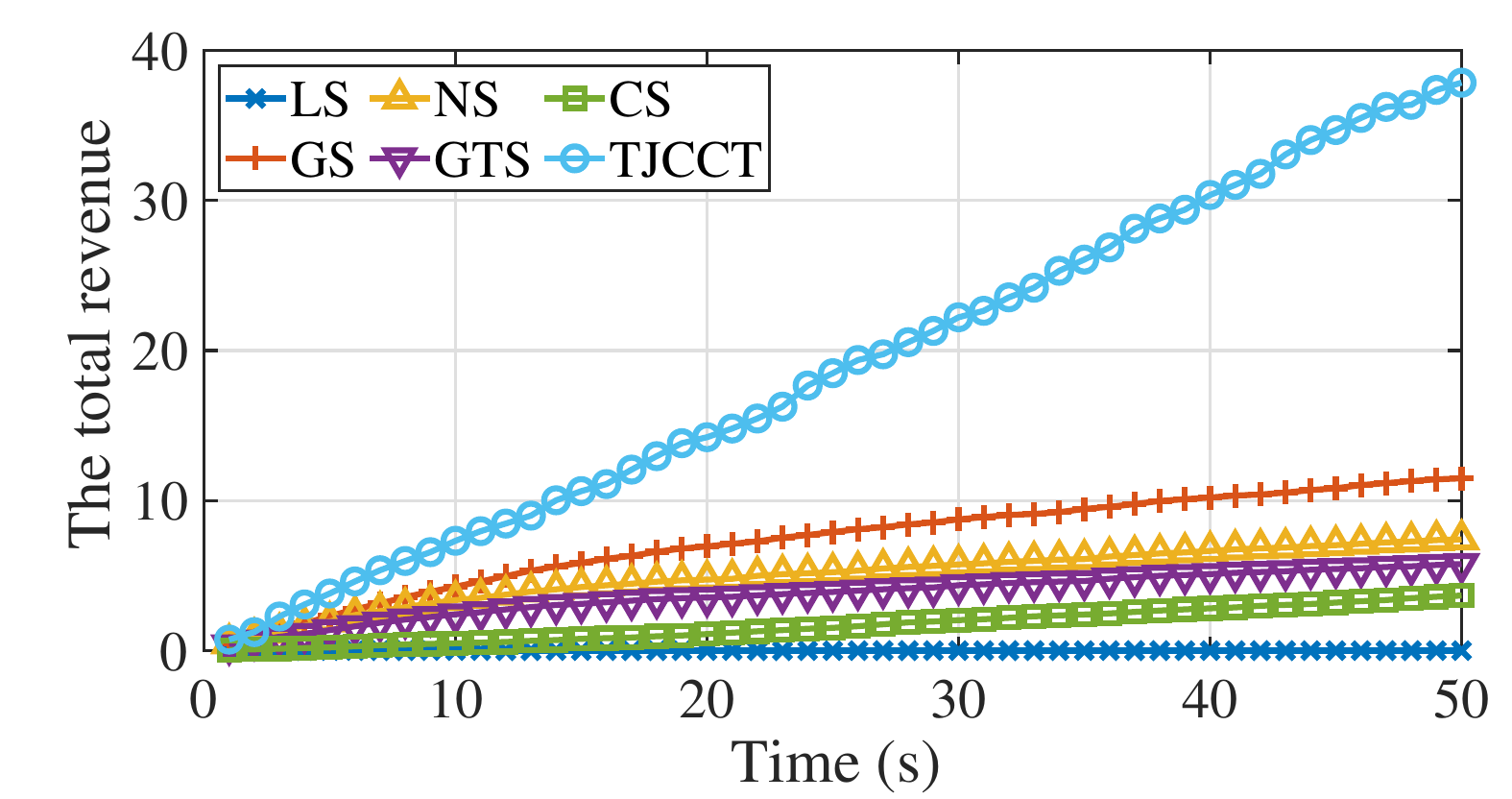}
		\end{minipage}
            \vspace{-1pt}
	}
	\centering
 \vspace{-1pt}
	\caption{{Effect of time. (a) The total utility of the system. (b) The aggregate QoE of MDs. (c) The total revenue of MEC servers.}}
	\label{fig_sw}
\end{figure*}

\begin{figure*}[!hbt] 
	\centering
		\setlength{\abovecaptionskip}{1pt}%
		\setlength{\belowcaptionskip}{1pt}%
	\subfigure[The total utility of the system]
	{
		\begin{minipage}[t]{0.31\linewidth}
			\raggedleft
			\includegraphics[scale=0.22]{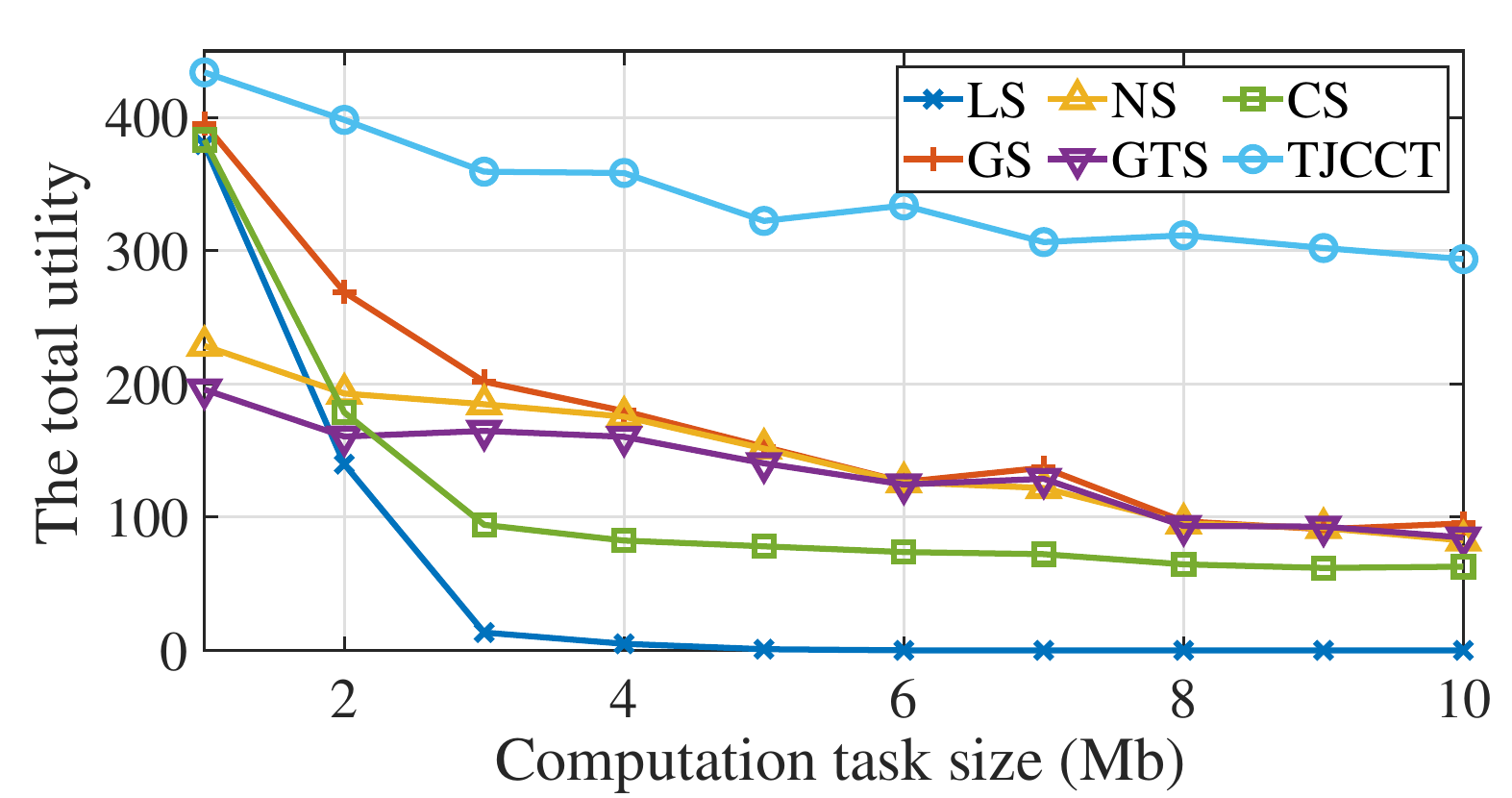}
		\end{minipage}
	}
	\subfigure[The aggregate QoE of MDs]
	{
		\begin{minipage}[t]{0.31\linewidth}
			\centering
			\includegraphics[scale=0.22]{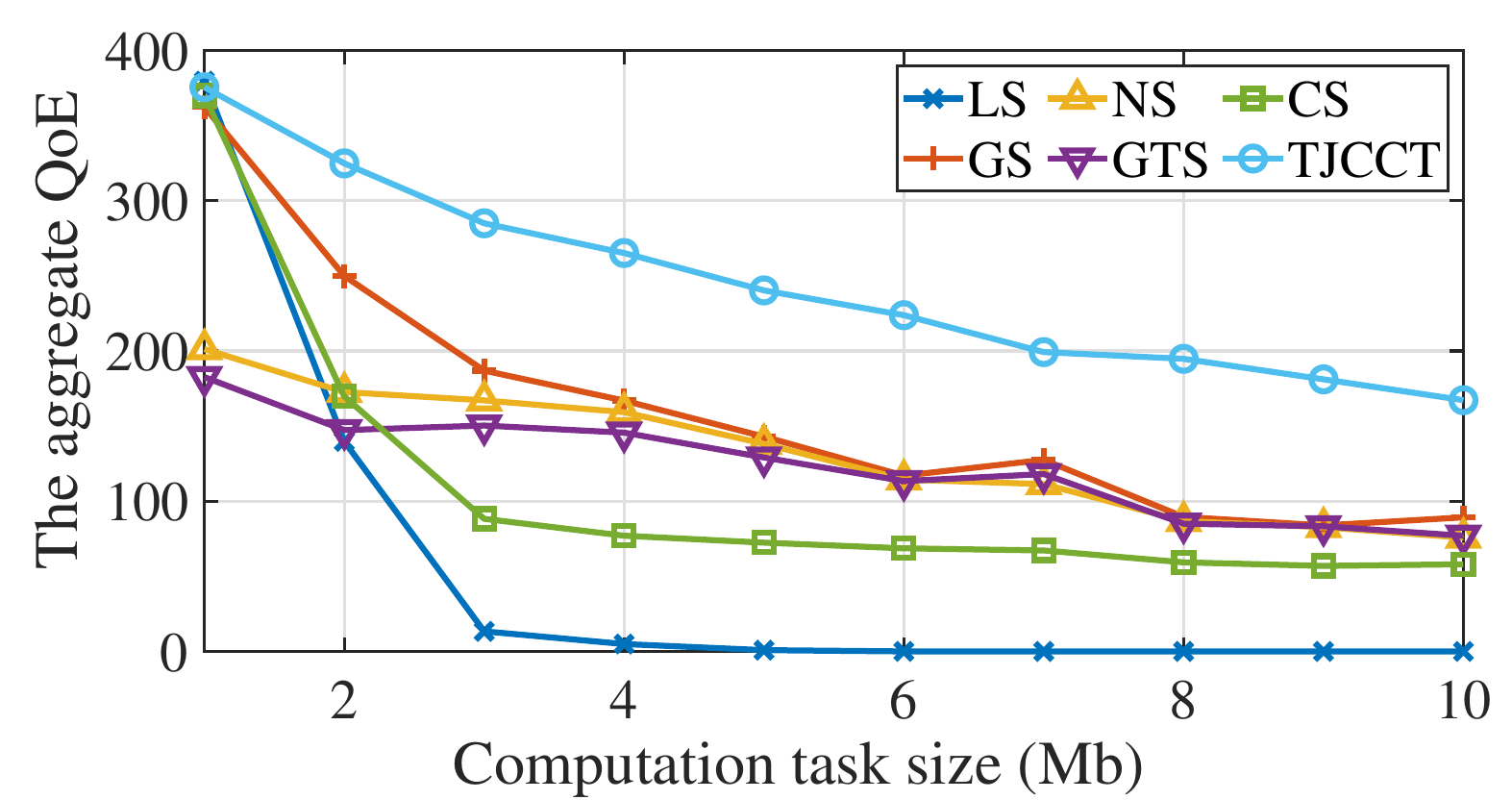}	
		\end{minipage}
	}
	\subfigure[The total revenue of MEC servers]
	{
		\begin{minipage}[t]{0.31\linewidth}
			\raggedright
			\includegraphics[scale=0.22]{utility-mec-eps-converted-to.pdf}
		\end{minipage}
	}
	\centering
	\caption{{Effect of the average computation size. (a) The total utility of the system. (b) The aggregate QoE of MDs. (c) The total revenue of MEC servers.}}
	\label{fig_task}
	\vspace{-1pt}
\end{figure*}

\par  \textbf{\textit{Effect of Average Computation Size.}} Figs. \ref{fig_task}(a), \ref{fig_task}(b), and \ref{fig_task}(c) show the impact of the average computation size on the total utility of the system, the aggregate QoE of MDs, and the total revenue of MEC servers, respectively with respect to the average computation size. From Figs. \ref{fig_task}(a) and \ref{fig_task}(b), it can be observed that TJCCT outperforms the comparative algorithms in the utility of the system and aggregate QoE of MDs, exhibiting a relatively gradual downward trend among the six algorithms with the increasing computation size. From Fig. \ref{fig_task}(c), it can be observed that TJCCT initially shows a slight upward trend and then approximately stabilizes as the workload increases, consistently maintaining a significantly superior level among the six schemes. The main reasons are as follows. On the one hand, due to the spatiotemporal-uneven distributed computation tasks, certain MEC servers could be rapidly overloaded by increasingly heavier workloads while that of the others remain idle. On the other hand, TJCCT shows superior capabilities to alleviate the possible congestion among MEC servers through the many-to-one matching scheme, to improve resource utilization through the on-demand computing resource allocation, and to provide dynamic offloading service to satisfy the QoE for MDs through trajectory control. Consequently, this set of results indicates that TJCCT outperforms the other algorithms in both light-loaded and heavy-loaded scenarios.


%
%

\section{Conclusion}
\label{sec_conclusion}

\par In this work, we study computing resource allocation, computation offloading, and UAV trajectory control for UAV-assisted MEC system. First, we employ a hierarchical framework to coordinate the collaboration among MDs, terrestrial edge, aerial edge, and the controller. Furthermore, we formulate an optimization problem to maximize the system utility. To solve the MINLP problem, we propose the TJCCT which consists of two time-scale optimization methods. In the short time scale, we propose a price-incentive method for computing resource allocation and a matching mechanism-based method for computation offloading. In the long time scale, we propose a convex optimization-based method for UAV trajectory control. Besides, the stability, optimality, and polynomial complexity of TJCCT are proved. Simulation results demonstrate that TJCCT outperforms the comparative algorithms is able to achieve long-term benefits in terms of the system utility, the QoE of MDs, and the revenue of MEC servers in both light-loaded and heavy-loaded scenarios.

\newpage
\bibliographystyle{IEEEtran}
\bibliography{references_1.bib}

\end{document}